\documentclass{article}
\usepackage{enumerate}
\usepackage{amssymb}
\usepackage{amsmath}
\usepackage{mathrsfs}
\usepackage{graphicx}
\usepackage{subfigure}
\usepackage{color}
\usepackage{amsthm}

\usepackage{mathtools}
\usepackage{bbm}
\usepackage{authblk}
\newtheorem{theorem}{Theorem}

\newcommand\blfootnote[1]{%
  \begingroup
  \renewcommand\thefootnote{}\footnote{#1}%
  \addtocounter{footnote}{-1}%
  \endgroup
}

\newcommand{\E}{\mathbb{E}}
\newcommand{\ind}[1]{\mathbbm{1}\!\left\{#1\right\}}
\newcommand{\cF}{\mathcal{F}}
\newcommand{\hI}{\hat{I}}
\newcommand{\df}[1]{\textnormal{\,d} \tau}
\newcommand{\bI}{\overline{I}}
\newcommand{\ba}{\overline{\alpha}}

\date{}
\allowdisplaybreaks[1]
\usepackage{geometry}
\begin{document}
\title{Change Rate Estimation and Optimal Freshness in Web Page Crawling\blfootnote{This paper has been accepted to the 13th EAI International Conference on Performance Evaluation Methodologies and Tools, VALUETOOLS'20, May 18--20, 2020, Tsukuba, Japan.
This is the author version of the paper.}}

\author[1]{Konstantin Avrachenkov }
\author[1]{Kishor Patil}
\author[2]{Gugan Thoppe}
\affil[1]{INRIA Sophia Antipolis, France 06902\footnote{email: k.avrachenkov@inria.fr, kishor.patil@inria.fr, gugan.thoppe@gmail.com}}
\affil[2]{Indian Institute of Science, Bengaluru, India 560012}
\maketitle

\begin{abstract}
For providing quick and accurate results, a search engine maintains a local snapshot of the entire web. And, to keep this local cache fresh, it employs a crawler for tracking changes across various web pages. However, finite bandwidth availability and server restrictions impose some constraints on the crawling frequency. Consequently, the ideal crawling rates are the ones that maximise the freshness of the local cache and also respect the above constraints.

Azar et al. \cite{Azar2018} recently proposed a tractable algorithm to solve this optimisation problem.  However, they assume the knowledge of the exact page change rates, which is unrealistic in practice. We address this issue here. Specifically, we provide two novel schemes for online estimation of page change rates. Both schemes only need partial information about the page change process, i.e., they only need to know if the page has changed or not since the last crawled instance. For both these schemes, we prove convergence and, also, derive their convergence rates. Finally, we provide some numerical experiments to compare the performance of our proposed estimators with the existing ones (e.g., MLE).
\end{abstract}

%\keywords{Change rate estimates, Web crawling, Maximum likelihood, Stochastic approximation, Poisson point process}

\section{Introduction}
The world wide web is gigantic: it has a lot of interconnected information and both the information and the connections keep changing. However, irrespective of the challenges arising out of  this, a user always expects a search engine to instantaneously provide accurate and up-to-date results. A search engine deals with this by maintaining a local cache of all the useful web pages and their links. As the freshness of this cache determines the quality of the search results, the search engine regularly updates it by employing a crawler (also referred to as a web spider or a web robot). The job of a crawler is (a) to access various web pages at certain frequencies so as to determine if any changes have happened to the content since the last crawled instance and (b) to update the local cache whenever there is a change. To understand the detailed working of crawlers, see \cite{heydon1999mercator, castillo2005effective, Kumar2016, Olston2010, Edwards2001}.

In general, a crawler has two constraints on how often it can access a page. The first one is due to limitations on the available bandwidth. The second one---also known as the politiness constraint---arises  when a server imposes limits on the crawl frequency. The latter implies that the crawler can not access pages on that server too often in a short amount of time. Such constraints cannot be ignored, since otherwise the server may forbid the crawler from all future accesses.

In summary, to identify the ideal rates for crawling different web pages, a search engine needs to solve the following optimisation problem: Maximise the freshness of the local database subject to constraints on the crawling frequency.

In the early variants of this problem, the freshness of each page was assumed to be  equally important \cite{cho2000synchronizing, Edwards2001}. In such cases, experimental evidence surprisingly shows that the uniform policy---crawl all pages at the same frequency irrespective of their change rates---is more or less the optimal crawling strategy.

Starting from the  pioneering work in \cite{cho2003effective}, however, the freshness definition was modified to include different weights for different pages depending on their importance, e.g., represented as the frequency of requests for the pages. The motivation for this change was the fact that only a finite number of pages can be crawled in any given time frame. Hence, to improve the utility of the local database, important pages should be kept as fresh as possible. Not surprisingly, under this new definition, the optimal crawling policy does indeed depend on the page change rates. This was numerically demonstrated first in  \cite{cho2003effective} for a setup with a small number of pages. A more rigorous derivation of this fact was recently given in the path breaking paper \cite{Azar2018} by Azar et al. In fact, this work also provides a near-linear time algorithm to find a near-optimal solution.

A separate study  \cite{avrachenkov2016whittle,nino2014dynamic} provides a Whittle index based dynamic programming approach to optimise the schedule of a web crawler. In that context, the page/catalogue freshness estimate also influences the optimal crawling policy and its good estimation is needed.

Our work is mainly motivated by the work from Azar et al. \cite{Azar2018}. In particular, input to their algorithm is the actual page change rates. However, in practice, these values are not known in advance and, instead, have to be estimated. This is the issue that we address in this paper.

Our main contributions can be summarised as follows. First, we propose two novel approaches for online estimation of the actual page change rates. The first is based on the Law of Large Numbers (LLN), while the second is derived using the Stochastic Approximation (SA) principles. Second, we theoretically show that both these estimators almost surely (a.s.) converge to the actual change rate values, i.e.,  both our estimators are asymptotically consistent. Furthermore, we also derive their convergence rates in the expected error sense. Finally, we provide some simulation results to compare the performance of our online schemes to each other and also to that of the (offline) MLE estimator. Alongside, we also show how our estimates can be combined with the algorithm in \cite{Azar2018} to obtain near-optimal crawling rates.

% Paper outlook
The rest of this paper is organised as follows. The next section provides a formal summary of this work in terms of the setup, goals, and key contributions. It also gives the explicit update rules for our two online schemes. In Section \ref{sec3}, we discuss their convergence and converge rates and also provide the formal analysis for the same. The numerical experiments discussed above are given in Section~\ref{sec5}. We conclude in Section~\ref{sec6} with some future directions.

\section{Setup, Goal, and Key Contributions}
\label{sec2}
The three topics are individually described below. \\[-1ex]

\noindent \textbf{Setup}: We assume the following. The local cache consists of copies of $N$ pages and $w_i$ denotes the  importance of the $i-$th page. Further, each page changes independently and the actual times at which page $i$ changes is a homogeneous Poisson point process in $[0, \infty)$ with a constant but unknown rate $\Delta_i.$ Independent of everything else, page $i$ is crawled (accessed) at the random instances $\{t_{k}\}_{k\geq 0} \subset [0, \infty),$ where $t_{0} = 0$ and the inter-arrival times, i.e., $\{t_{k} - t_{k - 1}\}_{k \geq 1},$ are iid exponential random variables with a known rate $p_i.$ Thus, the times at which page $i$ is crawled is also a Poisson point process but with rate $p_i.$ At time instance $t_{k},$ we get to know if page $i$ got modified or not in the interval $(t_{k - 1}, t_{k}],$ i.e., we can access the value of the indicator
\begin{equation*}
I_{k} := \begin{cases*}
                1, & if page i got modified in $(t_{k-1},t_{k}],$ \\
                0, & otherwise.
            \end{cases*}
\end{equation*}

We emphasise that each page is crawled independently. In other words, the notations $\{t_k\}$ and $\{I_k\}$ defined above do depend on $i.$ However, we hide this dependence for the sake of notational simplicity. We shall follow this practice for the other notations as well; the dependence on $i$ should be clear from the context.

Although the above assumptions are standard in the crawling literature, nevertheless, we now provide a quick justification for the same. Our assumption that the page change process is a Poisson point process is based on the experiments reported in \cite{brewington2000dynamic, brewington2000keeping, cho2000evolution}. Some generalised models for the page change process  have also been considered in the literature \cite{matloff2005, singh2007}; however, we do not pursue these ideas here. Separately, our assumption on $\{I_k\}$ is based on the fact that a crawler can only access incomplete knowledge about the page change process. In particular, a crawler does not know when and how many times a page has changed between two crawling instances. Instead, all it can track is the status of a page at each crawling instance and know if it has changed or not with respect to the previous access.  Sometimes, it is possible to also know the time at which the page was last modified \cite{castillo2005effective, cho2003estimating}, but we do not consider this case here. \\[-1ex]

\noindent \textbf{Goal}: Develop online algorithms for estimating $\Delta_i$ in the above setup. Subsequently, find optimal crawling rates $\{p_i^*\}$ so that the overall freshness of the local cache defined by
\begin{equation}\label{Fresh}
    \E\bigg[\dfrac{1}{T} \int\displaylimits_{0}^{T}\bigg( \sum_{i = 1}^{N}w_i \ind{\text{Fresh}(i,t)}\bigg) dt\bigg]
\end{equation}
is maximised subject to $\sum_{i = 1}^N p_i \leq B.$  Here,
$T > 0$ is some finite horizon, $B \geq 0$ is a bound on the overall crawling frequency, $\ind{}$ is the indicator, and \text{Fresh$(i, t)$} is the event that page $i$ is fresh at time $t,$ i.e., the local copy matches the actual page. \\[-1ex]

\noindent \textbf{Key Contributions}:
We present two online methods for estimating $\Delta_i,$ the first based on the LLN and the second based on SA.  If $\{x_k\}$ and $\{y_k\}$ denote the  iterates of these two methods, then their update rules are as shown below.
\begin{itemize}
    \item \emph{LLN Estimator}: For $k \geq 1,$
    \begin{equation}
        \label{eqn:LLN_Est}
        x_{k} = p_i \hI_k/(k + \alpha_k - \hI_k).
    \end{equation}
    Here, $\hI_k = \sum_{j = 1}^k I_j;$ hence, $\hI_k = \hI_{k - 1} + I_k.$ And, $\{\alpha_k\}$ is any positive sequence satisfying the conditions in Theorem~\ref{thm:LLN_Est}; e.g.,  $\{\alpha_k\}$ could be $\{1\},$ $\{\log k\},$ or $\{\sqrt{k}\}.$ \\[-1ex]

    \item \emph{SA Estimator}: For $k \geq 0$ and some initial value $y_0,$
    \begin{equation}
        \label{eqn:SA_Est}
        y_{k + 1} = y_k + \eta_k[I_{k + 1}(y_k + p_i) - y_k].
    \end{equation}
    Here, $\{\eta_k\}$ is any stepsize sequence that satisfies the conditions in Theorem~\ref{thm:SA_Est}. For example, $\{\eta_k\}$ could be $\{1/(k + 1)^\gamma\}$ for some $\gamma \in (0, 1].$
\end{itemize}

We call these methods online because the estimates can be updated on the fly as and when a new observation $I_k$ becomes available. This contrasts the MLE estimator in which one needs to start the calculation from scratch each time a new data point  arrives. Also, unlike MLE, our estimators are never unstable. See Section~\ref{subsec:Comp} for the complete details on this.

Our main results include the following. We show that both $\{x_k\}$ and $\{y_k\}$ converge to $\Delta_i$ a.s. Further, we show that
\begin{enumerate}
    \item $\E\|x_k - \Delta_i\| = O\left(\max\left\{k^{-1/2}, \alpha_k/k\right\}\right),$ and

    \item $\E\|y_k - \Delta_i\| = O(k^{-\gamma/2})$ if $\eta_k = (k + 1)^\gamma$ with $\gamma \in (0, 1).$
\end{enumerate}

Finally, we provide three numerical experiments for judging the strength of our two estimators. In the first one, we compare the performance of our estimators to each other and also to that of the Naive estimator and the MLE estimator described in \cite{cho2003estimating}. In the second one, we combine our estimates with the algorithm in \cite{Azar2018} and compute the optimal crawling rates. Subsequently, we use this to measure the overall freshness of the local cache. In the last and final experiment, we look at the behaviour of our estimators for different choices of the sequences $\{\alpha_k\}$ and $\{\eta_k\}.$

\section{Change rate estimation} \label{sec3}
Here, we provide a formal convergence and convergence rate analysis for our two estimators. Thereafter, we compare their behaviours to that of the estimators that already exist in the literature---the Naive estimator, the MLE estimator, and the Moment Matching (MM) estimator.

\subsection{LLN Estimator}
\label{subsecLLN_Est}

Our first aim here is to obtain a formula for $\E[I_1].$ We shall use this later to motivate the form of our LLN estimator.

Let $\tau_1 = t_1 - t_0 = t_1.$ Then, as per our assumptions in Section~\ref{sec2}, $\tau_1$ is an exponential random variable with rate $p_i.$ Also, $\E[I_1 | \tau_1 = \tau ] = 1- \exp{(-\Delta_i \tau)}.$
These two facts put together show that

\begin{equation}
\label{eqn:Exp_Ind}
   \E\big[I_1\big] = \Delta_i/(\Delta_i + p_i).
\end{equation}
This gives the desired formula for $\E[I_1].$

From this last calculation, we have
\begin{equation}
\label{e:Actual_Delta_i_Formula}
\Delta_i = p_i \E[I_1]/(1 -\E[I_1])
\end{equation}
Separately, because $\{I_k\}$ is an iid sequence and $\E|I_1| \leq 1$ , it follows from the strong law of large numbers that $\E\big[I_1\big] = \lim_{k \to \infty}  \sum_{j = 1}^k I_j /k \; \text{a.s.}$
Thus,
\[
\Delta_i = p_i \frac{\lim_{k \to \infty}\sum_{j = 1}^k I_j/k}{1 - \lim_{k \to \infty}\sum_{j = 1}^k I_j/k} \quad  \text{a.s.}
\]
Consequently, a natural estimator for  $\Delta_i$ is
\begin{equation}
    x_k' = p_i \frac{\sum_{j = 1}^k I_j/k}{1 - \sum_{j = 1}^kI_j/k} = p_i \frac{\hI_k}{k - \hI_k},
\end{equation}
where $\hI_k$ is as defined below \eqref{eqn:LLN_Est}.

Unfortunately, the above estimator faces an instability issue, i.e., $x'_k = \infty$ when $I_1, \ldots, I_k$ are all $1.$ To fix this, one can add a non-zero term in the denominator. The different choices then gives rise to the LLN estimator defined in \eqref{eqn:LLN_Est}.

The following result discusses the convergence and convergence rate of this estimator.

\begin{theorem}
\label{thm:LLN_Est}
Consider the estimator given in \eqref{eqn:LLN_Est} for some positive sequence $\{\alpha_k\}.$
\begin{enumerate}
    \item If \; $\lim_{k \to \infty} \alpha_k/k= 0,$ then  $\lim_{k \to \infty} x_k = \Delta_i \; \text{a.s.}$

    \item Additionally, if \; $\lim_{k \to \infty} \log(k/\alpha_k)/k = 0,$ then
    \[
    \E |x_k - \Delta_i| = O\left(\max\left\{k^{-1/2}, \alpha_k/k\right\}\right).
    \]
\end{enumerate}
\end{theorem}
\begin{proof}
Let $\mu = \E[I_1],$ $\bI_k = \hI_k/k,$ and $\ba_k = \alpha_k/k.$ Then, observe that \eqref{eqn:LLN_Est} can be rewritten as $x_k = p_i \bI_k/(1 + \ba_k - \bI_k).$
Now, $\lim_{k \to \infty} \bI_k = \mu$ a.s. and $\lim_{k \to \infty} \ba_k = 0.$ The first claim holds due to the strong law of large numbers, while the second one is true due to our assumption. Statement (1) is now easy to see.

We now derive Statement (2).  From \eqref{e:Actual_Delta_i_Formula}, we have
\[
|x_k - \Delta_i| =  \left|x_k - p_i \frac{\mu}{1 - \mu} \right| \leq p_i\left(|A_k| + |B_k|\right),
\]
where
\[
A_k = \frac{\bI_k}{\ba_k + 1 - \bI_k} -  \frac{\mu}{\ba_k + 1 - \mu} \]
and
\[
B_k = \frac{\mu}{\ba_k + 1 - \mu} - \frac{\mu}{1 - \mu}.
\]
Since $\alpha_k > 0$ and, hence, $\ba_k > 0,$
\[
|B_k| = \ba_k \frac{\mu}{(1 - \mu)(\ba_k + (1 - \mu))} \leq \ba_k \frac{\mu}{(1 - \mu)^2}.
\]
Similarly,
\[
|A_k| \leq \left(\frac{1 + \ba_k}{1 - \mu}\right) \left(\frac{|\bI_k - \mu|}{\ba_k + 1 - \bI_k}\right).
\]

Because we have assumed  $\ba_k \to 0,$ we get $\lim_{k \to \infty} \E[B_k] = 0.$  It remains to show $\lim_{k \to \infty} \E[A_k] = 0.$ Towards that, let $\{\delta_k\}$ be a positive sequence that we will pick later. Then,
\[
\E\left[\frac{|\bI_k - \mu|}{\ba_k + 1 - \bI_k}\right] \leq \E [C_k] + \E [D_k]
\]
where
\[
C_k = \frac{|\bI_k - \mu|}{\ba_k + 1 - \bI_k}\ind{\bI_k - \mu \leq \delta_k \mu}
\]
and
\[
D_k = \frac{|\bI_k - \mu|}{\ba_k + 1 - \bI_k}\ind{\bI_k - \mu \geq \delta_k \mu}.
\]
On the one hand,
\[
\E [C_k] \leq \frac{\E|\bI_k - \mu|}{\ba_k + 1 - (1 + \delta_k)\mu}  \leq \frac{\sqrt{\text{Var}[I_1]}}{\sqrt{k} (\ba_k + 1 - (1 + \delta_k) \mu)}.
\]
On the other hand, since  $|\bI_k - \mu| \leq 2$ and $1 - \bI_k \geq 0,$ it follows by applying the Chernoff bound that
\[
\E [D_k] \leq \frac{2}{\ba_k} \Pr \{\bI_k \geq (1 + \delta_k)\mu \} \leq \frac{2}{\ba_k} \exp\left(-k\delta_k^2  \mu/3\right).
\]

We now pick $\{\delta_k\}$ so that $\delta_k^2 = 6\log(1/\, \ba_k)/(k \mu)$ for all $k \geq 1.$ Then,  $\E [D_k] \leq 2 \ba_k.$ Now, due to our assumptions on $\{\alpha_k\},$ $\lim_{k\to \infty} \E [D_k] = 0.$ Similarly, $\lim_{k \to \infty}\delta_k = 0,$ whence it follows that $\lim_{k \to \infty} \E[C_k] = 0.$ These relations together then show that $\lim_{k \to \infty} \E[A_k] = 0.$

The desired result now follows.
\end{proof}

\subsection{SA Estimator}
Here, we use the theory of stochastic approximation to study the behaviour of our SA estimator.

\begin{theorem}
\label{thm:SA_Est}
Consider the estimator given in \eqref{eqn:SA_Est} for some positive stepsize sequence $\{\eta_k\}.$
\begin{enumerate}
    \item Suppose that $\sum_{k = 0}^\infty \eta_k = \infty$ and $\sum_{k = 0}^\infty \eta_k^2 < \infty.$ Then, $\lim_{k \to \infty} y_k = \Delta_i \, \text{a.s.}$

    \item Suppose that $\eta_k = 1/(k + 1)^\gamma$ with $\gamma \in (0, 1).$ Then,
    \[
        \E \|y_k - \Delta_i\| = O\left(k^{-\gamma/2}\right).
    \]
\end{enumerate}
\end{theorem}
\begin{proof}
For $k \geq 0,$ let $\cF_k :=  \sigma(\Delta^0, I_1, \ldots, I_k).$ Then, from \eqref{eqn:Exp_Ind} and the fact that $\{I_k\}$ is an iid sequence, we get
\[
\E[I_{k + 1}(y_k + p_i) - y_k |\cF_k] = \frac{\Delta_i}{\Delta_i + p_i}  (y_k + p_i) - y_k = h(y_k),
\]
where $h(y) = (\Delta_i - y)p_i/(\Delta_i + p_i).$ Hence, one can rewrite \eqref{eqn:SA_Est} as
\begin{equation}
\label{eqn:SAForm}
y_{k + 1} = y_k + \eta_k [h(y_k) + M_{k + 1}],
\end{equation}
where
\begin{align*}
    M_{k + 1} = {} & [I_{k+1}(y_k + p_i) - y_k] - h(y_k)\\
    = {} & \left[I_{k + 1} - \frac{\Delta_i}{\Delta_i + p_i}\right](y_k + p_i).
\end{align*}
Since $\E[M_{k + 1}|\cF_k] = 0$ for all $k \geq 0,$ $\{M_k\}$ is a martingale difference sequence. Consequently, \eqref{eqn:SAForm} is a classical SA algorithm whose limiting ODE is
\begin{equation}
\label{eqn:limODE}
    \dot{y}(t) = h(y).
\end{equation}

Now, Statement (1) follows from Corollary~4 and Theorem~7 in Chapters 2 and 3, respectively, of \cite{borkar2009stochastic}, provided we show that:
\begin{enumerate}[i.)]
    \item $h$ is a globally Lipschitz continuous function.

    \item $\Delta_i$ is an unique globally asymptotically stable equilibrium of \eqref{eqn:limODE}.

    \item $\sum_{k = 0}^\infty \eta_k = \infty$ and $\sum_{k = 0}^\infty \eta_k^2 < \infty.$

    \item $\{M_k\}$ is a martingale difference sequence with respect to the filtration $\{\cF_k\}.$ Further, there is a constant $C \geq 0$ such that $
        \E[M_{k + 1}^2|\cF_k] \leq C(1 + y_k^2) \quad \text{a.s.} $ for all $k \geq 0.$

    \item There exists a continuous function $h_\infty$ such that the functions $h_c(x) := h(cx)/c,$ $c \geq 1,$ satisfy $h_c(x) \to h_\infty(x)$ uniformly on compact sets as $c \to \infty.$

    \item The ODE $\dot{y}(t) = h_{\infty}(y)$ has origin as its unique globally asymptotically stable equilibrium.
\end{enumerate}

Since $h$ is linear, the Lipschitz continuity condition trivially holds. Separately, observe that $h(\Delta_i) = 0;$ this shows that $\Delta_i$ is an equilibrium point of \eqref{eqn:limODE}. Now,  $L(y) = (y - \Delta_i)^2/2$ is a Lyapunov function for \eqref{eqn:limODE} with respect to $\Delta_i.$ This is because $L(y) \geq 0,$ while $\nabla L(y) h(y) = -p(y - \Delta_i)^2/(p_i + \Delta_i) \leq 0;$ the equality holds in both these relations if and only if $y = \Delta_i.$ This shows that $\Delta_i$ is a unique globally asymptotically stable equilibrium of \eqref{eqn:limODE}, which  establishes Condition ii.).

Condition iii.) trivially holds due to our assumption about $\{\eta_k\}.$ Regarding the next condition, observe that $\{M_k\}$ is indeed a martingale difference sequence. Further,  $|M_{k + 1}| \leq |y_k| + p_i,$ whence it follows that Condition iv.) also holds.

Next, let $h_\infty(y) := -y p_i/(\Delta_i + p_i).$ Then, it is easy to see that Condition v.) trivially holds. Similarly, it is easy to see that Condition vi.) holds as well.

Statement (1) now follows, as desired.

We now sketch a proof for Statement (2). First, note that
\[
y_{k + 1} - \Delta_i = (1 - \lambda \eta_k) (y_k - \Delta_i) + \eta_k M_{k + 1},
\]
where $\lambda = p_i/(\Delta_i + p_i).$ Now, since $\E[M_{k + 1}|\cF_k] = 0,$
\[
\E[(y_{k + 1} - \Delta_i)^2|\cF_k] = (1 - \lambda \eta_k)^2(y_k - \Delta_i)^2 + \eta_k^2 \E[M_{k + 1}^2|\cF_k].
\]
Recall that $\E[M_{k + 1}^2 |\cF_k] \leq C(1 + y_k^2)$ for some constant $C \geq 0.$ Using this above and then repeating all the steps from the proof of  \cite[Theorem~3.1]{dalal2018finite} gives Statement (2), as desired.
\end{proof}

\subsection{Comparison with Existing Estimators}
\label{subsec:Comp}
As far as we know, there are three other approaches in the literature for estimating page change rates---the Naive estimator, the MLE estimator, and the MM estimator. The details about the first two estimators can be found in \cite{cho2003estimating} while, for the third one, one can look at \cite{upadhyay2019}. We now do a comparison, within the context of our setup, between these estimators and the ones that we have proposed.

The Naive estimator simply uses the average number of changes detected to approximate the rate at which a page changes. That is, if $\{z_k\}$ denote the values of the Naive estimator then, in our setup, $z_k = p_i\hI_k/k,$ where $\hI_k$ is as defined below in \eqref{eqn:LLN_Est}. The intuition behind this is the following. If  $\tau_1$ is as defined at the beginning of Section~\ref{subsecLLN_Est}, then observe that $\E[N(\tau_1)] = \Delta_i/p_i.$ Hence, the Naive estimator tries to approximate $\E[N(\tau_1)]$ with $\hI_k/k$ so that the previous relation can then be used for guessing the change rate.

Clearly, $ \E[z_k] = p_i \Delta_i/(\Delta_i + p_i) \neq \Delta_i.$ Also, from the strong law of large numbers, $z_k \overset{a.s.}{\to} p_i \Delta_i/(\Delta_i + p_i) \neq \Delta_i.$ Thus, this estimator is not consistent and is also biased. This is to be expected since this estimator does not account for all the changes that occur between two consecutive accesses.

Next, we look at the MLE estimator. Informally, this estimator identifies the parameter value that has the highest probability of producing the
observed set of observations.In our setup, the value of the MLE estimator is obtained by solving the following equation for $\Delta_i:$
\begin{equation}\label{LLE}
\sum_{j=1}^{k}  I_j\, \tau_j/(\exp{( \Delta_i\, \tau_j)} - 1) =  \sum_{j=1}^{k}  (1 - I_j)\, \tau_j,
\end{equation}
where $\tau_k = t_k - t_{k - 1}$ and $\{t_k\}$ is as defined in Section~\ref{sec2}.
The derivation of this relation is given in \cite[Appendix C]{cho2003estimating}. As mentioned in \cite[Section 4]{cho2003estimating}, the above estimator is consistent.

Note that the MLE estimator makes actual use of the inter-arrival crawl times $\{\tau_k\}$ unlike our two estimators and also the Naive estimator. In this sense, it fully accounts for the randomness in crawling intervals. And, as we shall see in the numerical section, the quality of the estimate obtained via MLE improves rapidly in comparison to the Naive estimator as the sample size increases.

However, MLE suffers in two aspects--- computational tractability and mathematical instability. Specifically, note that the MLE estimator lacks a closed form expression. Therefore, one has to solve \eqref{LLE} by using numerical methods such as the Newton–Raphson method, Fisher’s Scoring Method, etc. Unfortunately, using these ideas to solve \eqref{LLE} takes more and more time as the number of samples grow. Also note that, under the above solution ideas, the MLE estimator works in an offline fashion. In that, each time we get a new observation, \eqref{LLE} needs to be solved afresh. This is because there is no easy way to efficiently reuse the calculations from one iteration into the next. One reasonable alternative is to perform MLE estimation in a batch mode, i.e., wait until we gather a large number of samples and then apply one of the above-mentioned methods. However, even then the computation time will be long when $k$ is large.

Besides the complexity, the MLE estimator is also unstable in two situations. One, when no changes have been detected ($I_j = 0, \, \forall k \in \{1, \ldots, k\}$), and the other, when all the accesses detect a change ($I_j = 1, \, \forall k \in \{1, \ldots, k\}$). In the first setting, no solution exists; in the second setting, the solution is $\infty.$ One simple strategy to avoid these instability issues is to clip the estimate to some pre-defined range whenever one of bad observation instances occur.

Finally, we talk about the MM estimator. Here, one looks at the fraction of times no changes were detected during page accesses and, then, using a moment matching method tries to approximate the actual page change rate. In our context, the value of this estimator is obtained by solving $\sum_{j = 1}^k (1 - I_j) = \sum_{j = 1}^k e^{- \Delta_i \tau_j}$ for $\Delta_i.$ The details of this equation are given in \cite[Section~4]{upadhyay2019}. While the MM idea is indeed simpler than MLE, the associated estimation process continues to suffer from similar instability and computational issues like the ones discussed above.

We emphasise that none of our estimators suffer from any of the issues mentioned above. In particular, both our estimators are online and have a significantly simple update rule;  thus, improving the estimate whenever a new data point arrives is extremely easy. Also, both our estimators are stable, i.e., the estimated values will almost surely be finite. More importantly,  the performance of our estimators is comparable to that of MLE. This can be seen from the numerical experiments in Section~\ref{sec5}.

\begin{figure}
\centering
\subfigure[$\Delta_1 = 5,\quad p_1 = 3$] {\includegraphics[scale=0.36]{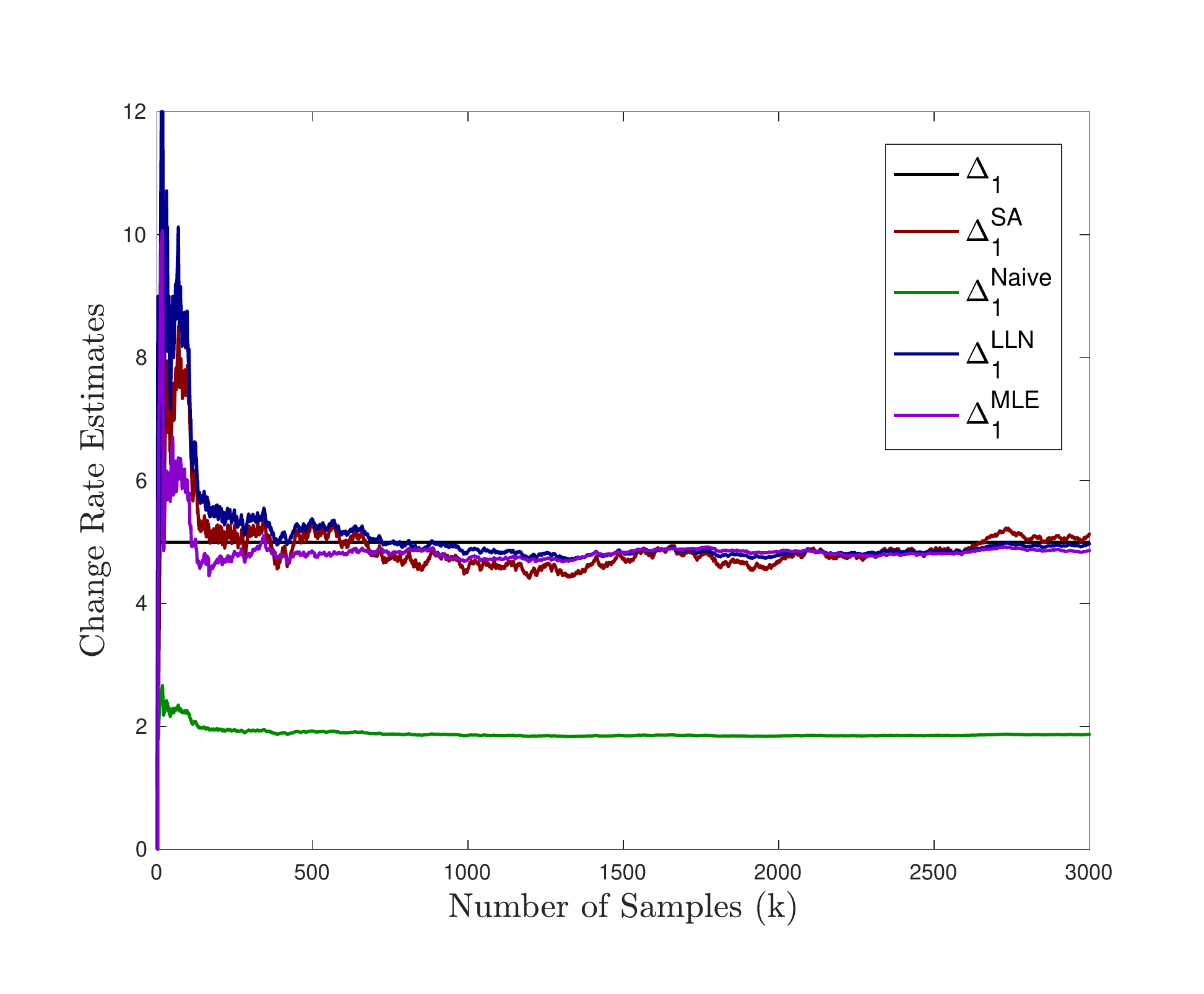}\label{page2}}\hspace{-2em}
\subfigure[$95\%$ Confidence interval] {\includegraphics[scale=0.36]{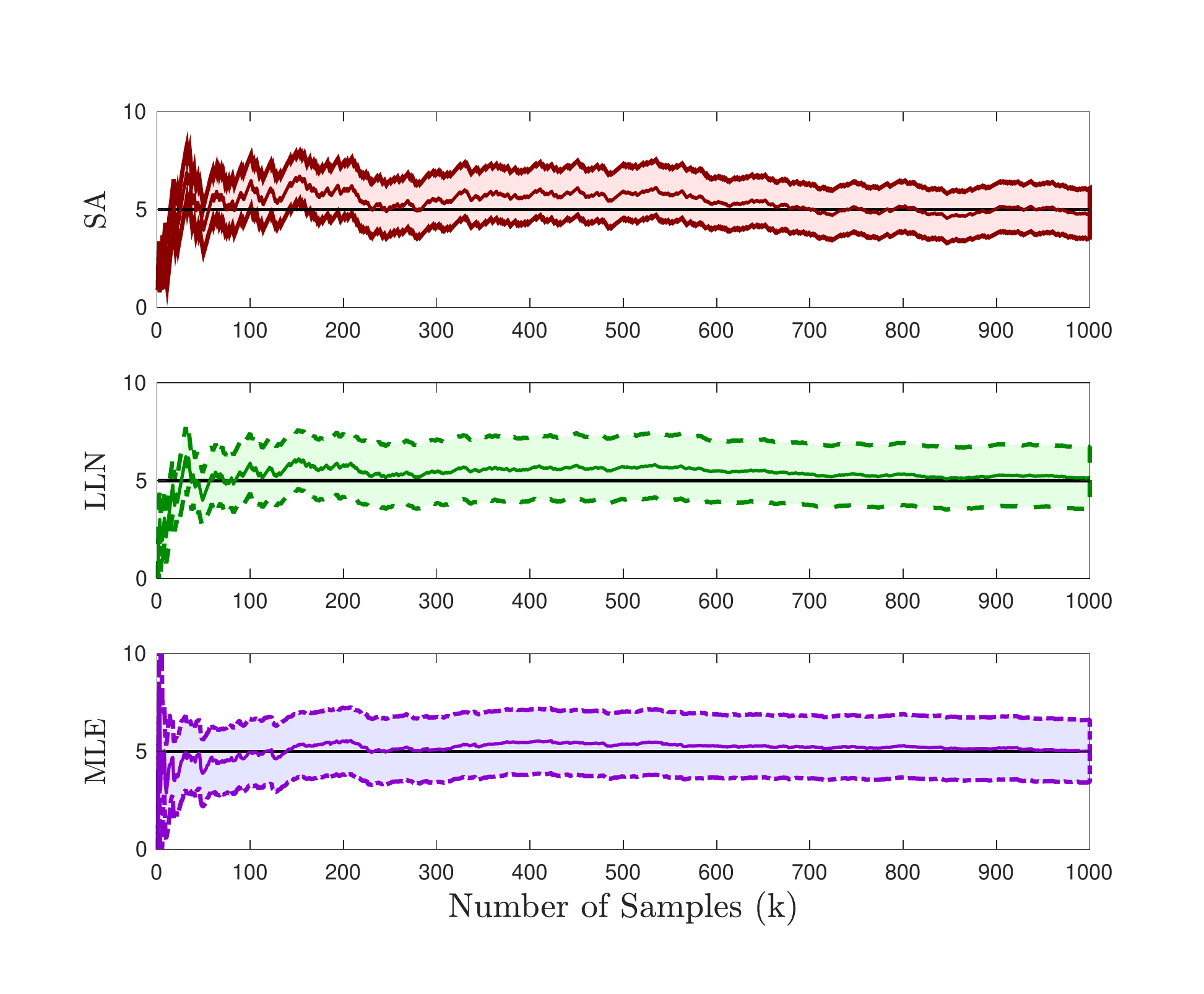}\label{CI}}
\caption{Comparison between Different Estimators.}
\label{fig: change_rates_comparison}
\end{figure}

\section{Numerical Results}
\label{sec5}
In this section, we provide three simulations to help evaluate the strength of our estimators. In the first experiment, we look at how well our estimation ideas perform in comparison to the Naive and the MLE estimator. In the second experiment, we substitute the change rate estimates obtained via the above approaches into the algorithm given in \cite{Azar2018} and compute the optimal crawling rates. To judge the quality of the crawling policy so obtained, we also look at the associated average freshness as defined in \eqref{Fresh}. Finally, in the third experiment, we compare the performance of our two estimators for different choices of $\{\alpha_k\}$ and $\{\eta_k\},$ respectively.

%We highlight at the outset that we work with synthetic data.

\subsection*{Expt. 1: Comparison of Estimation Quality}
Here, we compare four different page rate estimators: LLN, SA, Naive, and MLE. Their performances can be seen in Fig~\ref{fig: change_rates_comparison}. We now describe what is happening in the two figures there. Unless specified, the notations are as in Section~\ref{sec2}.

In Fig.~\ref{page2}, we work with exactly one page. We suppose that the times at which this page changes is a homogeneous Poisson point process with rate $\Delta_1 = 5.$ Separately, we set the crawling frequency arbitrarily to be $p_1 = 3.$ This implies that the times at which we crawl this page is another Poisson point process with rate $p_1 = 3.$

Using the above parameters, we now generate the random time instances at which this page changes. Alongside, we also sample the time instances at which this page is crawled. We then check if the page has changed or not between two successive page accesses. This generates the values of indicator sequence $\{I_k\}.$

We now give $\{I_k\},$ $\{\tau_k\},$ and $p_i$ as input to the four different estimators mentioned above and analyse their performances. The trajectory shown in Fig.~\ref{page2} corresponds to  exactly one run of each estimator. Note that the trajectory of the estimates obtained by the SA estimator is labelled $\Delta_1^{SA},$ etc. For the SA estimator, we had set $\eta_k = (k + 1)^{-\gamma}$ with $\gamma = 0.75.$ On the other hand, for our LLN estimator, we had set $\alpha_k \equiv 1.$

In Fig.~\ref{CI}, the parameter values are exactly in Fig.~\ref{page2}. However, we now run the simulation $1000$ times; the page change times and the page access times are generated afresh in each run. We then look at the $95\%$ confidence interval of the obtained estimates.

We now summarise our findings. Clearly, in each case, we can observe that performances of the MLE, LLN, and the SA estimators are comparable to each other and all of them outperform the Naive estimator. This last observation is not surprising since the Naive estimator completely ignores the missing changes between two crawling instances. However, the fact that the estimates from our approaches are close to that of the MLE estimator---both in terms of mean and variance---was indeed surprising to us. This is because, unlike MLE, our estimators completely ignore the actual lengths of the intervals between two accesses. Instead, they use $p_i,$ which only accounts for the mean interval length. %

While the plots do not show this, we once again draw attention to the fact that the time taken by each iteration in MLE rapidly grows as $k$ increases. However, our estimators take roughly the same amount of time for each iteration. %Nevertheless, they provide comparable performance in terms of mean and variance to that of MLE estimate.

\begin{figure}
\centering
\subfigure[Optimal crawling rate for Page $2$] {\includegraphics[scale=0.36]{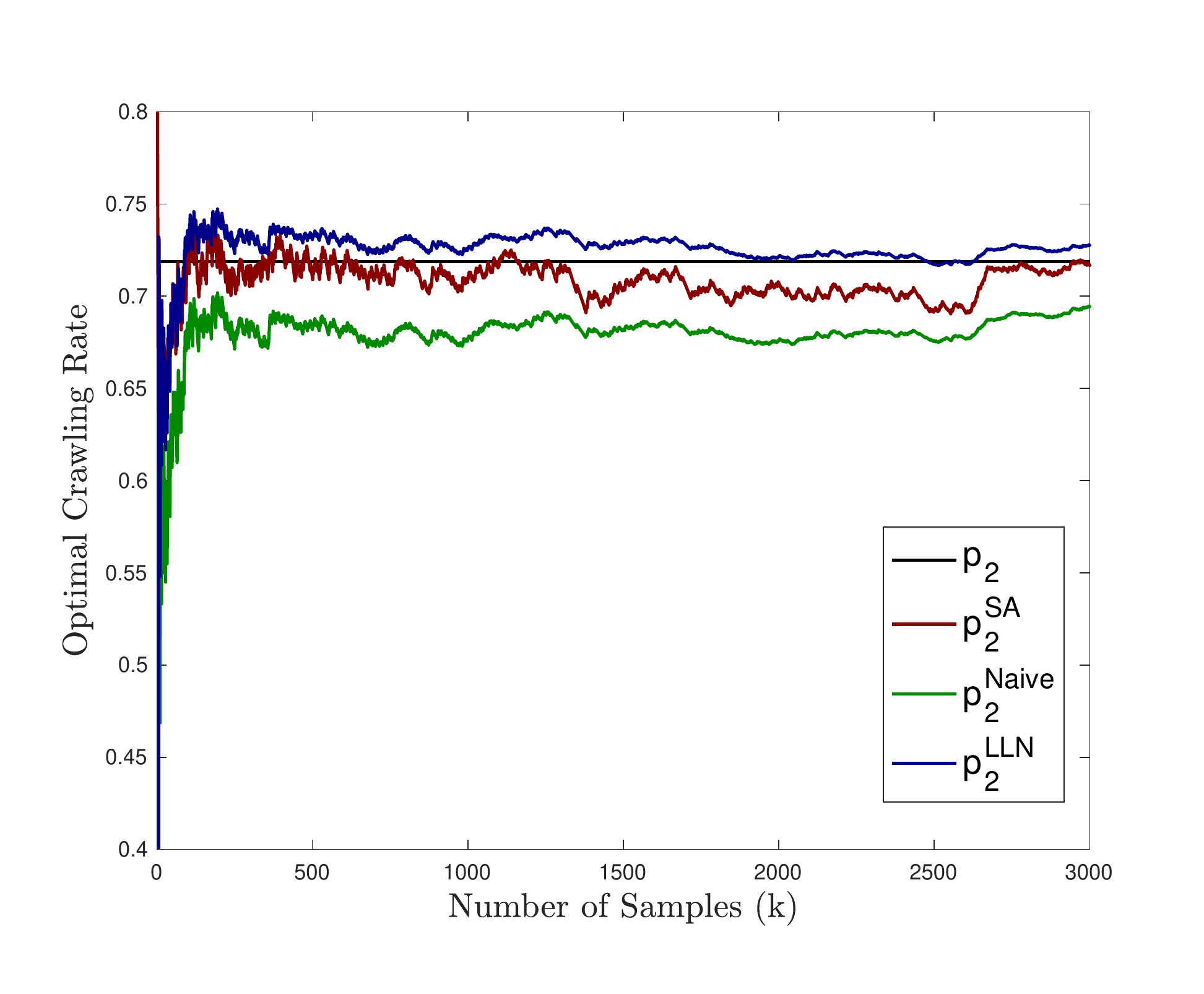}\label{FP}}\hspace{-2em}
\subfigure[Average Freshness] {\includegraphics[scale=0.36]{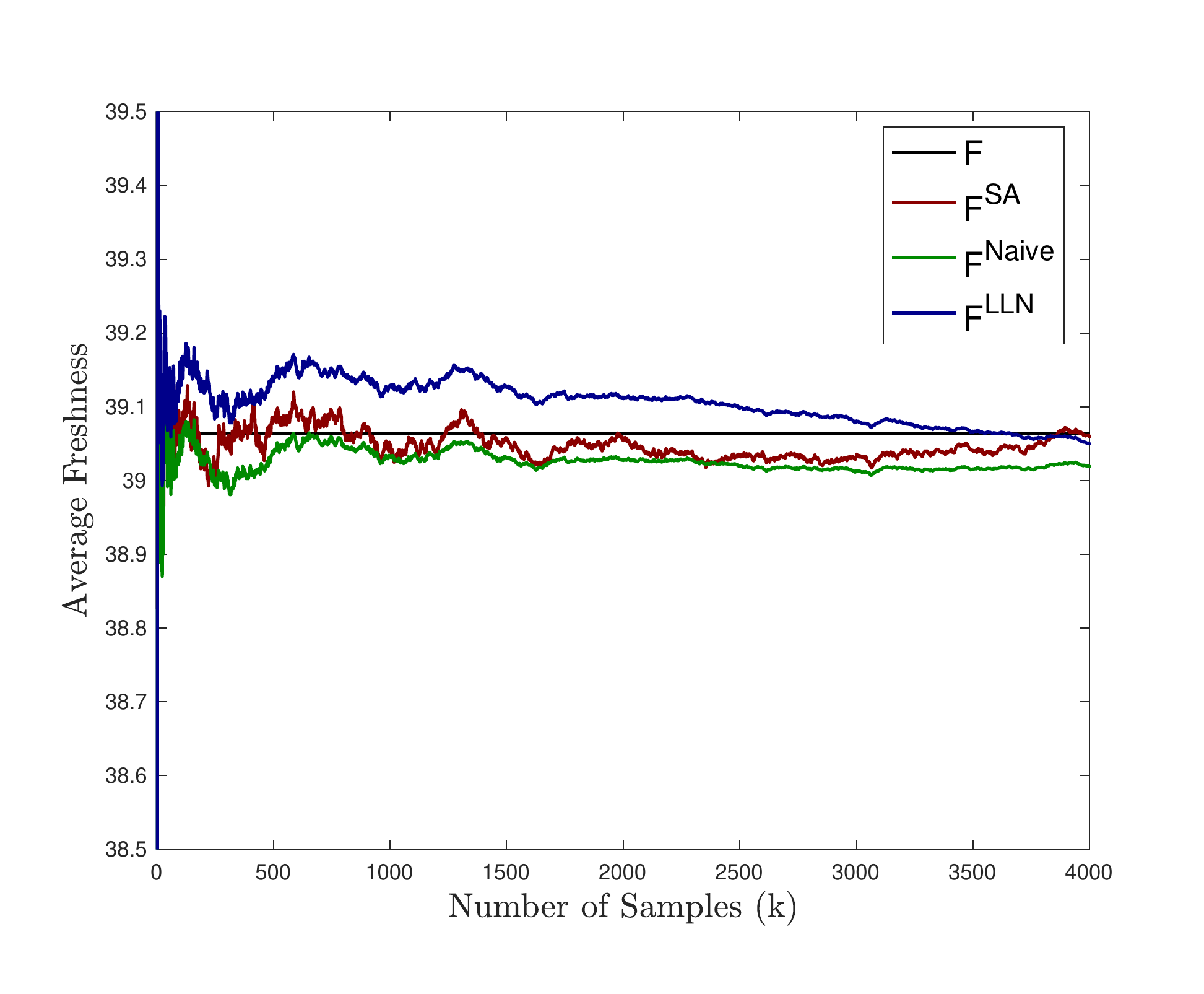}\label{TF}}
\caption{Optimal Crawling Rates and Freshness}
\label{figii}
\end{figure}

\subsection*{Expt. 2: Optimal Crawling rates and Freshness}

In this experiment, we consider $N = 100$ pages together. The  $\{\Delta_i\}$ sequence---the mean change rates for different pages---is obtained by sampling independently from the uniform distribution on $[0, 1],$ i.e., $\Delta_i \sim U[0,\, 1]$. We further assume that the bound on the overall bandwidth is $B = 80.$ The initial crawling frequencies for different pages are set by breaking up $B$ evenly across all pages, i.e., $p_i = B/N = 0.8$ for all $i.$ Because the $p_i$ values are arbitrarily chosen, these are not the optimal crawling rates. We then independently generate the change and access times for each page as in Expt. 1. Subsequently, we estimate the unknown change rate for each page individually.

For each $k,$ we then substitute the change rate estimates given by the different estimators into \cite[Algorithm 2]{Azar2018} and obtain the associated optimal crawling rates. In the same way, we substitute the actual $\Delta_i$ values there and obtain the true optimal crawling rates. Fig.~\ref{FP} provides a comparison between these values for a single page.  We can see that the estimate of the optimal crawling rate obtained from our approaches is much better than that of the Naive estimator.

To check how good our estimate of the true optimal crawling policy is, we look at the associated average freshness given by\footnote{In \cite{Azar2018}, it was shown that maximising \eqref{Fresh} under a bandwidth constraint for large enough $T$ corresponds to maximising \eqref{OptFresh} under the same bandwidth constraint.}
\begin{equation}\label{OptFresh}
    F(p) = \sum_{i=0}^{N}\dfrac{w_i p_i}{p_i+\Delta_i}
\end{equation}
and compare the same to that of the true optimal crawling policy. This comparison is given  in Fig.~\ref{TF}. Somewhat surprisingly, the average freshness does not vary much for all the three estimators. However, eventually, the average freshness captured by our estimators becomes much closer to the true optimal average freshness.

%Using the different approximations for the optimal crawling policy obtained above, we then compute the total average freshness as defined in \cite{Azar2018}.

\begin{figure*}
\centering
\subfigure[LLN estimator for different $\{\alpha_k\}$ choices] {\includegraphics[scale=0.36]{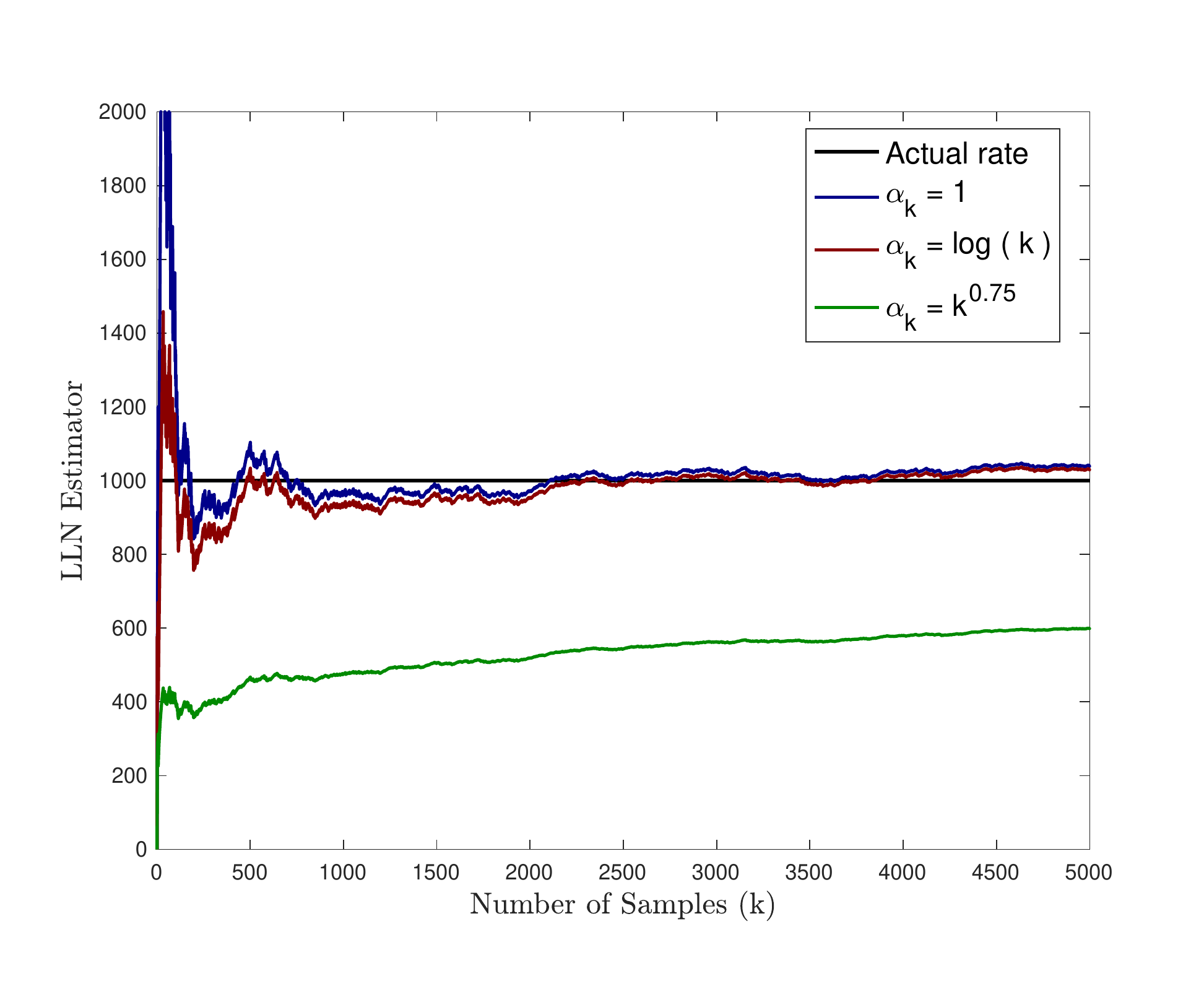}\label{Fig6}}\hspace{-2em}
\subfigure[SA estimator with $\eta_k = (k+1)^{-\gamma} $for different $\gamma$ choices] {\includegraphics[scale=0.36]{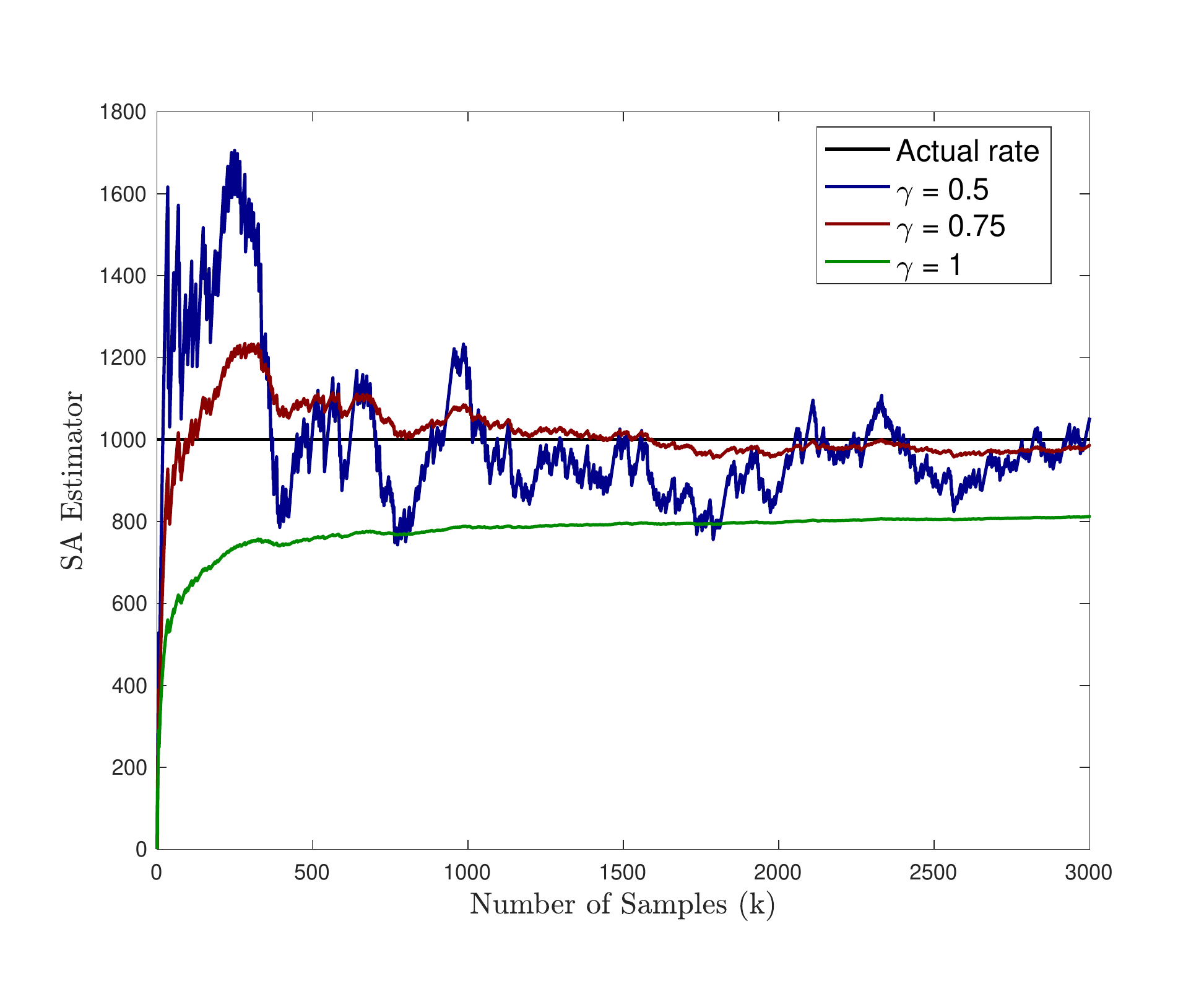}\label{Fig7}}
\caption{Impact of $\{\alpha_k\}$ and $\{\eta_k\}$ choices on Performance.}
\label{figiii}
\end{figure*}

\subsection*{Expt. 3: Impact of $\{\alpha_k\}$ and $\{\eta_k\}$ choices}
The theoretical results presented in Section~\ref{sec3} showed that the convergence rate of our estimators is affected by the choice of $\{\alpha_k\}$ and $\{\eta_k\},$ respectively. Figures~\ref{Fig6} and \ref{Fig7} provide a numerical verification of the same.

The details are as follows. Here, again, we restrict our attention to one single page. For Fig.~\ref{Fig6}, we chose $\Delta = 1000$ and $p = 200.$ Notice that the page change rate is very high, whereas the crawling frequency is relatively a low value. We then used the LLN estimator with three different choices of $\{\alpha_k\};$ these choices are shown in the figure itself. The LLN estimator with $\alpha_k = k^{0.75}$ has the worst performance. This behaviour matches the prediction made by Theorem~\ref{thm:LLN_Est}.

In Fig.~\ref{Fig7}, we again consider the same setup as above. However, this time we run the SA estimator with three different choices of $\{\eta_k\};$ the choices are given in the figure itself. We see that the performance for $\gamma = 0.75$ is better than the $\gamma = 0.5$ case. This is as predicted in Theorem~\ref{thm:SA_Est}. However, it worsens for the $\gamma = 1$ case. Notice that the latter case is not covered by Theorem~\ref{thm:SA_Est}.

\section{Conclusion and Future work} \label{sec6}
We proposed two new online approaches for estimating the rate of change of web pages. Both these estimators are computationally efficient in comparison to the MLE estimator. We first provide theoretical analysis on the convergence of our estimators and then provide numerical simulations to compare their performance with the existing estimators in the literature.  From numerical experiments, we have verified that the proposed estimators perform significantly better than the Naive estimator and have extremely simple update rules which make them computationally attractive.

The performance of both our estimators  currently depend on the choice of $\{\alpha_k\}$ and $\{\eta_k\},$ respectively. One aspect to analyse in the future would be to ask what would be the ideal choice for these sequences that would help attain the fastest convergence rate. Another interesting research direction is to combine the online estimation with dynamic optimisation.

\section*{Acknowledgement}
This work is partly supported by ANSWER project PIA FSN2 (P15 9564-266178 \textbackslash  DOS0060094) and DST-Inria  project  "Machine Learning for Network Analytics" IFC/DST-Inria-2016-01/448.

\bibliographystyle{acm}
\bibliography{acmart}
\end{document}